\documentclass[12pt]{article}%
\usepackage{amsfonts}
\usepackage{amsmath}
\usepackage{amssymb}
\usepackage{graphicx}%
\providecommand{\U}[1]{\protect\rule{.1in}{.1in}}
\newtheorem{theorem}{Theorem}

\newtheorem{condition}[theorem]{Condition}

\newtheorem{corollary}[theorem]{Corollary}

\newtheorem{definition}[theorem]{Definition}

\newtheorem{lemma}[theorem]{Lemma}

\newtheorem{proposition}[theorem]{Proposition}
\newtheorem{remark}[theorem]{Remark}

\newenvironment{proof}[1][Proof]{\textbf{#1.} }{\ \rule{0.5em}{0.5em}}

\begin{document}

\title{Graphical Exchange Mechanisms}
\author{Pradeep Dubey\thanks{Center for Game Theory, Department of Economics, Stony
Brook University; and Cowles Foundation for Research in Economics, Yale
University}, Siddhartha Sahi\thanks{Department of Mathematics, Rutgers
University, New Brunswick, New Jersey}, and Martin Shubik\thanks{ Cowles
Foundation for Research in Economics, Yale University; and Santa Fe Institute,
New Mexico.}}
\date{14 December 2015}
\maketitle

\section*{Abstract}

Consider an exchange mechanism which accepts \textquotedblleft
diversified\textquotedblright\ offers of various commodities and redistributes
everything it receives. We impose certain conditions of fairness and
convenience on such a mechanism and show that it admits unique prices, which
equalize the value of offers and returns for each individual.

We next define the complexity of a mechanism in terms of certain integers
$\tau_{ij},\pi_{ij}$ and $k_{i}$ that represent the time\ required to exchange
$i$ for $j$, the difficulty\ in determining the exchange ratio, and the
dimension\ of the message space. We show that there are a \emph{finite }number
of minimally complex mechanisms, in each of which all trade is conducted
through markets for commodity pairs.

Finally we consider minimal mechanisms with smallest worst-case\ complexities
$\tau=\max\tau_{ij}$ and $\pi=\max\pi_{ij}$. For $m>3$ commodities, there are
precisely three such mechanisms, one of which has a distinguished commodity --
the money -- that serves as the sole medium of exchange. As $m\rightarrow
\infty$ the money mechanism is the only one with bounded $\left(  \pi
,\tau\right)  $.

\textbf{JEL Classification}: C70, C72, C79, D44, D63, D82.

\textbf{Keywords: }exchange mechanism, minimal complexity, prices, markets, money.

\section{Introduction}

The purpose of this paper is to show how simple criteria of fairness,
convenience and complexity can lead to the successive emergence of prices,
markets, and money, in a Cournotian setting for commodity exchange. In the
process, we arrive at a rationale for money which is purely \textquotedblleft
mechanistic\textquotedblright\ in spirit, and complements the existing
utilitarian and behavioral literature on the subject (see \ref{Rel Lit}).

We consider abstract exchange mechanisms\footnote{For us a Cournotian
\textquotedblleft mechanism\textquotedblright\ is a formal device that enables
everyman to trade, with the simple expedient of offering commodities and
without having to account for his precise motivation or even bothering to
pretend that he has one. (See section 1.1.3.1 of \cite{Mertens: 2003} for a
spirited defence of the use of \textquotedblleft mechanism\textquotedblright%
\ in this plain English meaning of the word.) To forestall confusion, we
emphasize that our usage is different from that of the recent mechanism-design
literature, where the word has acquired a specialized connotation.}, which
accept \textquotedblleft diversified\textquotedblright\ offers of various
commodities and redistribute everything, and further satisfy five conditions,
embodying fairness and convenience that we term \emph{Anonymity, Aggregation,
Invariance, Non-dissipation }and\emph{ Flexibility.} Although there are
infinitely many such mechanisms, our first result is that each admits unique
\emph{prices}\footnote{i.e., consistent exchange-rates across commodity
pairs.}, which lead to \emph{value conservation, }\textit{i.e.} equalize the
value of individual offers and returns.

We next define some natural notions of \textquotedblleft
complexity\textquotedblright\ for a mechanism, and, in keeping with the idea
of convenience, study mechanisms with minimal complexity. This leads to a
\emph{finite}\textit{ }class $\mathfrak{M}_{\ast}\subset\mathfrak{M}_{g}$,
where $\mathfrak{M}_{g}$ denotes certain \emph{graphical\ mechanisms} that are
in one-to-one correspondence with directed, connected graphs on the set of
commodities. The directed edge $ij$ may be interpreted in $\mathfrak{M}_{\ast
}$ as a \emph{market}\ that provides traders the opportunity to offer $i$ in
exchange for $j$. Prices not only conserve values in $\mathfrak{M}_{\ast}$ (in
fact, in $\mathfrak{M}_{g}$) but \emph{mediate trade} in the sense that the
return to a trader depends only on his own offer and the prices. In short,
prices \textquotedblleft decouple\textquotedblright\ the interaction between traders.

The emergence of prices and markets paves the way for the culmination of the
analysis, namely the emergence of \emph{money}. To this end we introduce
additional refined criteria of complexity on $\mathfrak{M}_{\ast}$ and study
the corresponding minimal mechanisms, which we term \emph{strongly }minimal.
It turns out that there are only \emph{three} strongly minimal mechanisms, up
to a relabeling of commodities. In one of these, a single commodity emerges
endogenously as \emph{money} and mediates trade among decentralized markets
for the other commodities. Moreover, with a moderate increase in the number of
commodities, the money mechanism quickly supersedes the other two in a very
precise sense.

Finally let us mention that, while this paper is a companion to
\cite{Dubey-Sahi-Shubik}, the two are meant to be readable independently. This
has necessitated some overlap but it is minimal. To be precise, the conditions
on the mechanism (with the exception of Flexibility) appear in
\cite{Dubey-Sahi-Shubik}, so does Proposition \ref{Linearity}, and the rest is disjoint.

\subsection{Related Literature\label{Rel Lit}}

The emergence of money as a medium of exchange has been a matter of
considerable discussion in economics. One approach, following Jevons
\cite{Jevons:1875}, has focused on search-theoretic models that involve
repeated and random bilateral meeting between agents (see, \textit{e.g.},
\cite{Bannerjee-Maskin(1996)}, \cite{Iwai:1996}, \cite{Jones:1976},
\cite{Kiyotaki-Wright: 1989}, \cite{Kiyotaki-Wright: 1993},
\cite{Li-Wright:1998}, \cite{Ostroy:1973}, \cite{Trejos-Wright: 1995} and the
references therein). Another line of inquiry is based on partial or general
equilibrium models with various kinds of frictions in trade, such as
transactions costs or limited trading opportunities (see, \textit{e.g.},
\cite{Foley:1970}, \cite{Hahn:1971}, \cite{Heller: 1974}, \cite{heller-Starr:
1976}, \cite{Howitt-Clower: 2000}, \cite{Ostroy-Starr: 1974},
\cite{Ostroy-Starr: 1990}, \cite{starr:2012}, \cite{Starret:1973},
\cite{Wallace: 1980}). These models turn on notions of rational expectations
and utility-maximizing behavior of the agents in equilibrium. In contrast, as
was said, our analysis is based purely on the mechanism of trade is and
independent of the characteristics of the agents, such as their endowments or utilities.

It is worth emphasizing that our analysis is quite agnostic regarding the
choice of any particular money\footnote{Our model can equally accomodate fiat
money or commodity money, depending on how preferences are introduced. Indeed,
\emph{all} we suppose is that the $m$ items being traded are distinct from
each other. In particular, offers and returns could just be quotes (think of
e-commerce!), instead of actual shipment of goods.}, being only at pains to
point out the urgency of appointing \emph{some} money.\ For a discussion of
different criteria entailed in the choice of a suitable \textquotedblleft
commodity money\textquotedblright\ such as its portability, verifiability,
divisibility and durability; or, alternatively, the backing of the state
requisite to sustain \textquotedblleft fiat money\textquotedblright, see,
\emph{e,g}., \cite{Adam Smith}, \cite{Jevons:1875}, \cite{Knapp:1905},
\cite{Kualla:1920}, \cite{Lerner:1947}, \cite{Menger:1892}, \cite{Schumpeter:
1954}; and, for a recent survey on both kinds of money, see \cite{Shubik:1993}
and \cite{starr:2012}. It would be interesting to incorporate some of these
criteria, as well as the utilitarian considerations for money, into our
mechanistic framework.

This paper is intimately related to \cite{Dubey-Sahi-Shubik}. Let us briefly
recount the model there. We make the \emph{hypothesis} in
\cite{Dubey-Sahi-Shubik} that any offer of a commodity $i$ specifies some
other commodity $j$ which is being sought in exchange for $i.$ Thus, drawing a
directed arc $ij$ for every such offer permitted in the mechanism, we obtain a
(directed, connected\footnote{connected, because we require that it should be
possible to convert any commodity to another by iterated trading.}) graph $G$.
There are infinitely many mechanisms for any given $G,$ but it is shown in
\cite{Dubey-Sahi-Shubik} that exactly one of them is categorically determined
by \emph{four} of our five conditions, namely \emph{Anonymity, Aggregation,
Invariance, Non-dissipation }(all except \emph{Flexibility})\emph{. }This is
the graphical \textquotedblleft$G$-mechanism\textquotedblright\ mentioned
earlier; so that the class $\mathfrak{M}_{g}$ is precisely the one generated
as $G$ varies over all directed, connected graphs on the fixed node-set of
commodities. Our refined complexity criteria apply equally to $\mathfrak{M}%
_{g}$ as to its subset $\mathfrak{M}_{\ast}$, and we show that $\mathfrak{M}%
_{g}$ has the \emph{same} three strongly minimal mechanisms as $\mathfrak{M}%
_{\ast}.$ This is the main conclusion of \cite{Dubey-Sahi-Shubik} and it
constitutes a key step in the proof of the emergence of money in this paper.

Our current paper thus puts the analysis of \cite{Dubey-Sahi-Shubik}\ on a
more general footing. We start with a domain which is much richer than
$\mathfrak{M}_{g}$, and which permits traders to indulge in \textquotedblleft
cheap talk\textquotedblright\ in order to diversify their offers in any
commodity $i.$ Furthermore, we allow for the possibility that these offers may
get complicated \emph{bundles} of commodities in return. But
\emph{Flexibility} guarantees that there exist special messages in the cheap
talk through which a trader can \textquotedblleft unbundle\textquotedblright%
\ his return, i.e., messages that enable him to get \emph{only} $j$ in
exchange for $i$, whenever $j$ is obtainable --- albeit in conjunction with
other commodities --- via an offer of $i$. Thus flexibility \emph{embeds} a
graph $G,$ as a sharply delineated language, within the tangle of cheap talk.
We essentially show that our complexity criteria cuts away the tangle, leaving
behind only the graphical $G$-mechanisms in $\mathfrak{M}_{\ast}%
\subset\mathfrak{M}_{g}$. \emph{Thus we deduce the existence of markets (i.e.,
the edges of the graph }$G$\emph{) from an abstract standpoint, instead of
postulating them as in }\cite{Dubey-Sahi-Shubik}\emph{. }The exact structure
of the subdomain $\mathfrak{M}_{\ast}$ of $\mathfrak{M}_{g}$ is not explored
by us, as that would detract from the primary purpose of this paper, which is
to arrive expeditiously at the money mechanism.

The other precursor of this paper, at both the technical and conceptual level,
is \cite{Dubey-Sahi: 2003}, where a mechanism produces not only trades but
also prices, based upon everyone's offers; and where \emph{price mediation }is
\emph{postulated }as a condition. The main result of \cite{Dubey-Sahi: 2003}
is that $\mathfrak{M}_{g}$ is characterized by \emph{Anonymity, Aggregation,
Invariance, Price Mediation }and \emph{Accessibility }(the last representing a
weak form of continuity)$.$ In contrast, here we deduce the existence of
prices, as well as their crucial role in mediating trade, based on
considerations of a different sort, as outlined in the introduction and
formalized in section \ref{G-mechanisms} below. Moreover, the analysis in
\cite{Dubey-Sahi: 2003} stopped at the characterization of $\mathfrak{M}_{g}$
and did not delve into any further selection among the mechanisms.

\section{Exchange Mechanisms\label{Formal model}}

\begin{definition}
A \emph{pre-mechanism} consists of the following data:

\begin{enumerate}
\item a commodity space $C=\mathbb{R}_{+}^{m}$;

\item an action space $S=\mathbb{R}_{+}^{K}$ where $K=K_{1}\amalg\cdots\amalg
K_{m}$ is a finite\emph{ set;}

\item for each integer $n\geq1$, a \emph{return} map $\rho^{n}:S\left(
n\right)  \rightarrow C^{n}$ where
\begin{equation}
S\left(  n\right)  =\left\{  \left(  a^{1},\ldots,a^{n}\right)  \in
S^{n}:a^{1}+\cdots+a^{n}\in S_{+}\right\}  ,\quad S_{+}=\mathbb{R}_{++}^{K}.
\label{=pos-agg}%
\end{equation}

\end{enumerate}
\end{definition}

We refer to the elements of $K_{i}$ as $i$\emph{-indices}, and for $a\in S$ we
define $\overline{a}\in C$ by summing over the various $i$-indices for each
$i$; thus we have%
\[
\overline{a}=\left(  \alpha_{1},\ldots,\alpha_{m}\right)  \text{ where }%
\alpha_{i}=%
{\textstyle\sum\nolimits_{h\in K_{i}}}
a_{h}\text{.}%
\]

\begin{definition}
An \emph{exchange mechanism} is a pre-mechanism which satisfies:
\begin{equation}
\text{if }\left(  r^{1},\ldots,r^{n}\right)  =\rho^{n}\left(  a^{1}%
,\ldots,a^{n}\right)  \text{ then }r^{1}+\cdots+r^{n}=\overline{a^{1}}%
+\cdots+\overline{a^{n}}. \label{=cc}%
\end{equation}

\end{definition}

The interpretation is as follows. There is an underlying set $\left\{
1,\ldots,m\right\}  $ of commodities\footnote{It will always be clear from the
context whether $i$ is the name of a commodity, or that of an individual, or
just an integer.}, and for each $i$-index $h\in K_{i}$, the component $a_{h}$
of an action $a$ represents an offer of commodity $i.$ Thus we also refer to
$a$ as an \emph{offer}, and the various $i$-indices serve to \textquotedblleft
diversify\textquotedblright\ the offer in $i$. An exchange mechanism enables
trade as follows: having received $n$ offers $\left(  a^{1},\ldots
,a^{n}\right)  $, which are \textquotedblleft positive on the
aggregate\textquotedblright\ (\ref{=pos-agg}), it redistributes the
commodities according to $\rho^{n}$. Condition (\ref{=cc}) means that
commodities are conserved.

As mentioned already in section \ref{Rel Lit}, one possible interpretation of
the indices is as a common language in which traders may communicate with the
mechanism $M$. The language is completely abstract: no structure is imposed on
it except that it be of finite size. The elements of $K_{i}$ may be thought of
as costless messages (\textquotedblleft cheap talk\textquotedblright) that
accompany offers in $i$. Another interpretation of $K_{i}$ could be that it
represents different times (or, places) when (or, where) the offer of $i$ is
sent. The reader can no doubt think of still more interpretations. It is to
make room for all these that we have used the neutral term \textquotedblleft
index\textquotedblright. However, in the context of graphical mechanisms which
emerge out of our analysis (see section \ref{G-mechanisms}), it turns out that
$i$-indices have a concrete economic interpretation as certain commodity pairs
$ij$, representing \textquotedblleft markets\textquotedblright\ where $i$ can
be exchanged for $j$.

\subsection{Conditions on the Mechanisms\label{n:vec}}

In order to describe our conditions, we need some notation. We first define
scaling actions of $\lambda\in\mathbb{R}_{++}^{m}$ on $r\in C$ and $a\in S$
via%
\[
\left(  \lambda\ast r\right)  _{i}=\lambda_{i}r_{i}\text{ for all
}i\text{,\quad}\left(  \lambda\ast a\right)  _{h}=\lambda_{i}a_{h}\text{ for
all }h\in K_{i}\text{.}%
\]
Let $\rho:S\left(  2\right)  \rightarrow C$ denote the first component of the
two trader return function $\rho^{2}$. If $h\in K_{i}$ is an $i$-index then
the unit vector $e_{h}\in S$ is an offer in commodity $i$ alone. The $j$-th
component $\rho_{j}\left(  e_{h},a\right)  $ of $\rho\left(  e_{h},a\right)  $
is the return of commodity $j$ to trader 1, when he offers $e_{h}$ and the
other offers $a$.

\begin{definition}
We say $h\in K_{i}$ is an $ij$\emph{-index} if for some $a$ we have%
\[
\rho_{j}\left(  e_{h},a\right)  >0\text{.}%
\]
We say an $ij$-index $h$ is \emph{pure} if for all $a$ we have
\[
\rho_{k}\left(  e_{h},a\right)  =0\text{ for all }k\neq j\text{.}%
\]

\end{definition}

Our five conditions on an exchange mechanism, termed \emph{Anonymity},
\emph{Aggregation, Invariance, Non-dissipation,} and \emph{Flexibility, }are
as follows.

\begin{condition}
\label{Condition} If $\rho^{n}\left(  a^{1},\ldots,a^{n}\right)  =\left(
r^{1},\ldots,r^{n}\right)  $ then we have

\begin{enumerate}
\item \label{Anonymity} $\rho^{n}\left(  a^{\sigma\left(  1\right)  }%
,\ldots,a^{\sigma\left(  n\right)  }\right)  =\left(  r^{\sigma\left(
1\right)  },\ldots,r^{\sigma\left(  n\right)  }\right)  $ for any permutation
$\sigma.$

\item \label{Aggregation} $\rho^{n-1}\left(  a^{1},\ldots,a^{n-2}%
,a^{n-1}+a^{n}\right)  =\left(  r^{1},\ldots,r^{n-2},r^{n-1}+r^{n}\right)  $.

\item \label{Invariance} $\rho^{n}\left(  \lambda\ast a^{1},\ldots,\lambda\ast
a^{n}\right)  =\left(  \lambda\ast r^{1},\ldots,\lambda\ast r^{n}\right)  $
for all $\lambda\in\mathbb{R}_{++}^{m}$.

\item \label{Non-dissipation} For each $i,$ $r^{i}-\overline{a^{i}}$ is either
$0$ or it has a strictly positive component.

\item \label{Flexibility} For all $ij$, if there is an $ij$-index then there
is a pure $ij$-index.
\end{enumerate}
\end{condition}

For such a mechanism $M,$ by \emph{Anonymity} and \emph{Aggregation }we have
\[
\rho^{n}\left(  a^{1},\ldots,a^{n}\right)  =\left(  \rho\left(  a^{1}%
,a^{-1}\right)  ,\ldots,\rho\left(  a^{n},a^{-n}\right)  \right)  \text{ where
}a^{-i}=a^{1}+\cdots+a^{n}-a^{i}.
\]
Thus $M$ is uniquely determined by $\rho$, and also by the \emph{trade} and
\emph{net trade} functions, which are defined as follows:
\[
r\left(  a,b\right)  =\rho\left(  a,b-a\right)  ,\quad\nu\left(  a,b\right)
=r\left(  a,b\right)  -\bar{a}\text{.}%
\]
These latter functions have domain $\left\{  \left(  a,b\right)  \in S\times
S_{+}:a\leq b\right\}  $.

\begin{proposition}
\label{Linearity} $\nu$ admits a unique extension to $S\times S_{+}$
satisfying
\[
\nu\left(  ta,b\right)  =t\nu\left(  a,b\right)  ,\qquad\nu\left(
a,tb\right)  =\nu\left(  a,b\right)  \text{ for all }t>0.
\]

\end{proposition}

\begin{proof}
See Lemma 1 of \cite{Dubey-Sahi: 2003}. Although \cite{Dubey-Sahi: 2003}
considers a more restrictive class of mechanisms, we note that the proof of
Lemma 1 there only uses \emph{Anonymity, Aggregation, }and\emph{ Invariance.}
\end{proof}

In view of the above result, we will drop the restriction $a\leq b$ for
$\nu\left(  a,b\right)  $.

\subsection{Comments on the Conditions}

The first condition is\emph{ Anonymity; }it stipulates that the mechanism be
blind to all characteristics of a trader other than his offer. The second
condition is \emph{Aggregation}, it asserts that if any trader pretends to be
two different persons by splitting his offer, the returns to the others is
unaffected. \emph{Aggregation} does not imply that if two individuals were to
merge, they would be unable to enhance their \textquotedblleft
oligopolistic\ power\textquotedblright. For despite the aggregation condition,
the merged individuals are free to \emph{coordinate} their actions by jointly
picking a point in the Cartesian product of their action spaces. Indeed all
the mechanisms we obtain display this \textquotedblleft oligopolistic
effect\textquotedblright, even though they also satisfy \emph{Aggregation}.
The two conditions embody fairness, enabling free entry for any new
participant on non-discriminatory terms, and thereby making the mechanism more
\textquotedblleft inclusive\textquotedblright. They also contribute to
convenience, if either of these conditions were violated, trade would become a
cumbersome affair: each individual would need to keep track of the full
distribution of offers across the entire population, and then figure out how
to diversify his own offers in response.

The third condition is \emph{Invariance; }its main content is that the
\textit{maps} $\rho^{n}$ which comprise the mechanism are invariant under a
change of units in which commodities are measured. This makes the mechanism
much simpler to operate in: one does not need to keep track of seven pounds or
seven kilograms or seven tons, just the numeral $7$ will do. It is worthy of
note that the cuneiform tablets of ancient Sumeria, which are some of the
earliest examples of written language\ and arithmetic, are in large part
devoted to records and receipts pertaining to economic transactions.
\emph{Invariance} postulates the "numericity" property of the $\rho^{n}$
making them independent of the underlying choice of units, and this goes to
the very heart of the quantitative measurement of commodities. In its absence,
one would need to figure out how the maps are altered when units change, as
they are prone to do, especially in a dynamic economy. This would make the
mechanism cumbersome to use.

The fourth condition\ is \emph{Non-dissipation}; it says that no trader's
return can be less commodity-wise than his offer. If it were violated, such
unfortunate traders would find it grossly unfair and tend to abandon the
mechanism. In conjunction with \emph{Aggregation, Anonymity}, and the
conservation of commodities, this immediately implies \emph{no-arbitrage: }%
\[
\text{for any }a,b\text{ neither }\nu(a,b)\gvertneqq0\text{ nor }%
\nu(a,b)\lvertneqq0.
\]
To see this, note that in view of Proposition \ref{Linearity} we need consider
only the case $a\leq b$ and rule out $\nu(a,b)\gvertneqq0$. Denote $c=b-a.$
Then
\[
\nu(a,b)+\nu(c,b)=\nu(a+c,b)=\nu(b,b)=0,
\]
where the first equality follows from \emph{Aggregation}, and the last from
conservation of commodities. But then $\nu(a,b)\gvertneqq0$ implies
$\nu(c,b)\lvertneqq0,$ contradicting \emph{Non-dissipation.}

The fifth and final condition is \emph{Flexibility}, it reflects the
perspective of a trader\ who wishes to interact with the mechanism to
exchange\ a single commodity $i$ for some other commodity $j$. If $h$ is an
$ij$-index then we have $\rho_{j}\left(  e_{h},a\right)  >0$, which means that
the trader can get a positive amount of commodity $j$ for a suitable offer by
the other(s). However if there is no pure $ij$-index then the trader may be
forced to accept commodity $j$ bundled with other commodities.
\emph{Flexibility} guarantees that there are \textquotedblleft
enough\textquotedblright\ pure indices to enable individuals to
\textquotedblleft unbundle\textquotedblright\ their returns. The mechanism may
well admit complex trading opportunities, such as swaps of commodity bundles,
that coexist with these indices; the former comprising, so to speak, a tangled
web around the latter. It is our complexity criteria below which eliminate the
web and allow only the pure $ij$-indices to survive (as markets of $i$ for
$j$), see Theorem \ref{Emergence of G-mechanisms}.

\section{Complexity}

We now discuss three notions of \emph{complexity} for a mechanism $M$. The
first, and simplest, is
\[
k_{i}=k_{i}(M)=\left\vert K_{i}(M)\right\vert
\]
which is the dimension of the offer space for commodity $i$, and which we
refer to as the $i$-\emph{index complexity}.

The next two notions are defined from standpoint of a \textquotedblleft
binary\textquotedblright\ $ij$-trader\footnote{We focus on bilateral trades
between pairs of commodities because they form an iterative basis for all
trade. This is so on account of \emph{prices} (exchange rates) that will
shortly be shown to emerge and govern all trade.} who interfaces with $M$\ in
order to exchange commodity $i$ for \emph{exclusively }commodity $j$. We focus
on two basic concerns for such a trader: first, \emph{how long} will it take
him to effect the exchange; and, second, \emph{how difficult }will it be for
him to figure out the terms of exchange? The first concern leads to the notion
of \textquotedblleft time complexity\textquotedblright, and the second to that
of \textquotedblleft price complexity\textquotedblright.

We fix some notation; let $e_{i}$ denote the $i$-th unit vector.

\begin{definition}
A vector $v$ is an $i$\emph{-vector} if $v=se_{i}$ for some real number $s>0$;
and an $\bar{\imath}j$\emph{-vector} if $v=-se_{i}+te_{j}$ for some real
$s,t>0$.
\end{definition}

\subsection{Time Complexity\label{time-complexity}}

\begin{definition}
Given two commodity bundles $v,w\in C$ we will say that $v$ can be
\emph{converted} to $w,$ and we write $v\rightarrow w$ if there exist $a,b$
such that
\[
w=v+\nu(a,b)\text{ and }\overline{a}\leq v\text{.}%
\]

\end{definition}

Let $\tau\left(  v,w,M\right)  $ denote the smallest \textquotedblleft
time\textquotedblright\ $t$ for which there is a sequence\footnote{Notice that
at the moment we permit the market state to vary across the $t$ transitions of
the sequence $v\rightarrow v^{1}\rightarrow\cdots\rightarrow v^{t-1}%
\rightarrow w.$ But even if we were to restrict attention to the case in which
the \emph{same }$b\in S_{+}$ should represent the market state across all
transitions, the time complexity would of $M$ will remain unaffected. This
follows from Lemma \ref{lem:vw}}%
\[
v\rightarrow v^{1}\rightarrow\cdots\rightarrow v^{t-1}\rightarrow w.
\]
We define the $ij$-\emph{time complexity} and (maximum) \emph{time complexity}
as follows\footnote{It follows by \emph{Invariance} that in the definition one
can replace the unit vectors $e_{i},e_{j}$ by arbitrary $i$- and $j$-
vectors.}%
\[
\tau_{ij}\left(  M\right)  :=\tau\left(  e_{i},e_{j},M\right)  ,\quad
\tau(M):=\max\nolimits_{i\neq j}\left\{  \tau_{ij}(M)\right\}  .
\]
We say that a mechanism $M$ is \emph{connected }if $\tau(M)<\infty.$

\begin{definition}
\label{all mech} $\mathfrak{M(}m\mathfrak{)}$ is the class of \emph{all}
connected mechanisms with commodity set $\left\{  1,\ldots,m\right\}  $, that
satisfy \emph{Anonymity, Aggregation, Invariance, Non-dissipation }and
\emph{Flexibility.}
\end{definition}

When the commodity set $\left\{  1,\ldots,m\right\}  $ is understood, we shall
often suppress $m$ and write $\mathfrak{M=M}(m).$

\subsection{The Emergence of Prices}

Let $\mathbb{R}_{++}^{m}/\sim$ be the set of rays in $\mathbb{R}_{++}^{m},$
representing prices\footnote{A ray $p$ represents a price vector up to overall
multiplication by a positive scalar; the ratios $p_{i}/p_{j}$ represent
well-defined consistent exchange rates across all pairs $ij$ of commodities.}.
It turns out that prices\footnote{Our analysis remains intact if there is a
continuum of traders (see Section 7 of \cite{Dubey-Sahi-Shubik}). In this
case, an individual's action has no effect on the exchange rate. Otherwise it
affects the aggregate offer (i.e., the state of the market) and thereby the
exchange rate, which is but to be expected in an oligopolistic framework.}
emerge in connected mechanisms; and the values, under these prices, of offers
and returns are conserved for every trader.

\begin{theorem}
\label{Emergence of Prices} Let $M\in\mathfrak{M}$ with associated net trade
function $\nu.$ Then there is a unique map $p:$ $\mathbb{R}_{++}%
^{K}\rightarrow$ $\mathbb{R}_{++}^{m}/\sim$ satisfying \emph{value
conservation\footnote{Note that value conservation is perforce true on the
aggregate since commodities are neither created nor destroyed by the
mechanism, only redistributed. What is shown here is that it holds at the
individual level, i.e., the mechanism does not assign \textquotedblleft
profitable\textquotedblright\ trades to some at the expense of others.}:
}$p(b)\cdot\nu(a,b)=0.$
\end{theorem}

Even though $p(b)$ is only defined up to an overall scalar multiple, for each
pair $i,j$ we get a well-defined price ratio function
\[
p_{ij}:S_{+}\mapsto\mathbb{R}_{++};\qquad p_{ij}(b)=\frac{p_{i}(b)}{p_{j}(b)}%
\]
Theorem \ref{Emergence of Prices} has the following immediate consequence.

\begin{corollary}
Suppose $\nu(a,b)$ is an $\bar{\imath}j$-vector. Then $\frac{\nu_{i}(a,b)}%
{\nu_{j}(a,b)}=-p_{ij}(b).$
\end{corollary}

\subsection{Price Complexity\label{price complexity}}

Note that a binary $ij$-trader is only interested in net trades $\nu(a,b)$
that are $\bar{\imath}j$-vectors. By the previous corollary, the exchange
ratio $\frac{\nu_{i}(a,b)}{\nu_{j}(a,b)}$ is independent of the action $a$
producing the $\bar{\imath}j$-trade, and depends only on $p_{ij}(b).$
Therefore such a trader is interested only in those components of $b$ which
\textquotedblleft influence\textquotedblright\ the function $p_{ij}(b)$.

To make this notion precise, say that component $i$ is \emph{influential} for
a function $f(x_{1,}\ldots,x_{l})$ if there are two inputs $x,x^{\prime}$,
differing only in the $i$th place, such that $f\left(  x\right)  \neq f\left(
x^{\prime}\right)  $. Define the $ij$-\emph{price complexity }$\pi
_{ij}(M\mathcal{)}$ to be the number of influential components of the function
$p_{ij}.$ Also define the (maximum)\emph{ price complexity }by%
\[
\pi(M\mathcal{)}:=\max\left\{  \pi_{ij}(M\mathcal{)}:i\neq j\right\}
\]

\section{The Emergence of Markets: G-Mechanisms\label{G-mechanisms}}

\subsection{Directed Graphs}

In this paper by a graph we mean a \emph{directed simple} \emph{graph}. Such a
graph $G$ consists of a finite \emph{vertex} set $V_{G}$, togther with an
\emph{edge} set $E_{G}\subseteq V_{G}\times V_{G}$ that does not contain any
loops, \textit{i.e.,} edges of the form $ii$. For simplicity we shall often
write $i\in G$, $ij\in G$ in place of $i\in V_{G}$, $ij\in E_{G}$ but there
should be no confusion.

By a \emph{path} $ii_{1}i_{2}\ldots i_{k}j$ from $i$ to $j$ we mean a nonempty
sequence of edges in $G$ of the form%
\[
ii_{1},i_{1}i_{2},\ldots,i_{k-1}i_{k},i_{k}j.
\]
If $k=0$ then the path consists of the single edge $ij$, otherwise we insist
that the \emph{intermediate }vertices $i_{1},\ldots,i_{k}$ be distinct from
each other and from the endpoints $i,j$. However we do allow $i=j$, in which
case the path is called a \emph{cycle}. We say that $G$ is \emph{connected} if
for any two vertices $i\neq j$ there is a \emph{path} from $i$ to $j$.

\subsection{G-mechanisms\label{G-mech}}

Let $G$ be a connected graph with vertex set $\left\{  1,\ldots,m\right\}  $.
Following \cite{Dubey-Sahi: 2003} one may associate to $G$ a mechanism $M_{G}$
$\in\mathfrak{M=M}(m)$ as follows. We let $K_{i}$ be the set of outgoing edges
at vertex $i$, and regard $v\in S$ as a matrix $\left(  v_{ij}\right)  $ with
$v_{ij}$ understood to be $0$ if $ij\notin G$. To define $r\left(  a,b\right)
$ we need the following elementary result (see,\textit{ e.g.}
\cite{Dubey-Sahi: 2003}).

\begin{lemma}
For $b\in S_{+},$ there is a unique ray $p=p(b)$ in $\mathbb{R}_{++}^{m}/\sim$
satisfying
\begin{equation}
\sum\nolimits_{i}p_{i}b_{ij}=\sum\nolimits_{i}p_{j}b_{ji}\text{ for all }j.
\label{=price}%
\end{equation}

\end{lemma}

Now for $\left(  a,b\right)  \in S\times S_{+}$ we set $p=p\left(  b\right)  $
as in (\ref{=price}) and define $r\left(  a,b\right)  $ by
\begin{equation}
r_{i}(a,b)=p_{i}^{-1}\left(  \sum\nolimits_{j}p_{j}a_{ji}\right)  \text{ for
all }i. \label{MGr}%
\end{equation}
We remark that the left side of (\ref{=price}) is the total value of all the
goods \textquotedblleft chasing\textquotedblright\ good $j$, while the right
side is the total value of good $j$ on offer.

Mechanisms of the form $M_{G}$ will be called (connected) $G$%
\emph{-mechanisms}, and we write $\mathfrak{M}_{g}=$ $\mathfrak{M}_{g}(m)$ for
the totality of such mechanisms. It is worth noting that $\mathfrak{M}_{g}$ is
a \emph{finite} set. Moreover, the formula (\ref{MGr}) for the return function
of a $G$-mechanism immediately implies
\begin{equation}
p(b)=p(c)\Longrightarrow r(a,b)=r(a,c)\text{ for all }a\in S;b,c\in S_{+}
\label{pm}%
\end{equation}
In \cite{Dubey-Sahi: 2003} this property was referred to as \emph{price
mediation} and, in conjunction with other axioms, shown to characterize
$\mathfrak{M}_{g}.$

To sum up, these graphical $G$-mechanisms\ have very special structure. All
the indices are pure, \textit{i.e.} each edge $ij$ of $G$ represents a pure
$ij$-index and can be interpreted as a \emph{market }to exchange $i$ for $j$;
furthermore, as we just saw, prices mediate trade in $M_{G}$ in the following
strong sense: the return to a trader depends only on his own offer and the
prices\footnote{Price mediation in fact follows from value conservation
\emph{once }all indices are pure.}((see equation (\ref{pm})). Thus prices play
the full-fledged role of a \textquotedblleft decoupling
device\textquotedblright\ in any $G$-mechanism.

It is worth emphasizing that the markets of $G$-mechanisms are, in general,
\emph{not }decentralized in that the exchange rate $p_{i}/p_{j}$ may depend on
offers of commodities \emph{other} than $i$ and $j,$ at various edges in the graph.

\subsection{Minimal Mechanisms\label{min-mech}}

Given $M$ and $M^{\prime}$ in $\mathfrak{M=M(}m\mathfrak{)}$ with complexities
$\tau_{ij},\pi_{ij},k_{i}$ and $\tau_{ij}^{\prime},\pi_{ij}^{\prime}%
,k_{i}^{\prime}$ respectively, we say that $M$ is \emph{no more complex} than
$M^{\prime}$ and write $M\preceq M^{\prime}$ if for all $i,j$
\[
\tau_{ij}\leq\tau_{ij}^{\prime},\quad\text{ }\pi_{ij}\leq\pi_{ij}^{\prime
},\quad\text{ }k_{i}\leq k_{i}^{\prime}.
\]
Clearly $\preceq$ is reflexive and transitive, and hence constitutes a
quasiorder on $\mathfrak{M}$. We let $\mathfrak{M}_{\ast}=\mathfrak{M}_{\ast
}(m)$ denote the set of $\preceq$-minimal elements of $\mathfrak{M.}$

\begin{theorem}
\label{Emergence of G-mechanisms} Minimal mechanisms are $G$-mechanisms:
$\mathfrak{M}_{\ast}\subset\mathfrak{M}_{g}$.
\end{theorem}

\section{The Emergence of Money\label{Emergence -money}}

Let us, from now on, identify two mechanisms if one can be obtained from the
other by relabeling commodities. There are three mechanisms of special
interest to us in $\mathfrak{M}_{g}(m)$ called the \emph{star, cycle, }and
\emph{complete mechanisms; }with the following edge-sets:
\[%
\begin{tabular}
[c]{|c|c|c|c|}\hline
$G$ & Star & Cycle & Complete\\\hline
$E_{G}$ & $\left\{  mi,im:i<m\right\}  $ & $\left\{  12,23,\ldots,m1\right\}
$ & $\left\{  ij:i\neq j\right\}  $\\\hline
\end{tabular}
\ \ \ \ \ \ \ \ \ \ \ \
\]
\emph{ }Notice that the central vertex $m$ of the graph underlying the star
mechanism plays the role of money, and is the sole medium of
exchange\footnote{This is reminiscent of \textquotedblleft spontaneous
symmetry breaking\textquotedblright\ in physics. The \textit{ex ante }symmetry
between commodities, assumed in our model, is carried over to the cycle and
complete mechanisms. It breaks down only in the star mechanism, giving rise to
money.}.

Although the set $\mathfrak{M}_{\ast}$ is finite, it can be quite large and we
will not attempt to characterize it here. Instead we consider the
\textquotedblleft worst-case complexities\textquotedblright\ $\pi\left(
M\right)  =\max\pi_{ij}\left(  M\right)  $ and $\tau\left(  M\right)
=\max\tau_{ij}\left(  M\right)  $, and the corresponding quasiorder on
$\mathfrak{M,}$ namely: $M\preceq_{w}M^{\prime}$ if
\[
\tau(M\mathcal{)\leq\tau(}M^{\prime}),\quad\pi(M\mathcal{)\leq\pi(}M^{\prime
})\text{.}%
\]
If $\mathfrak{\tilde{M}}$ is a subset of $\mathfrak{M}$ one can consider the
minimal elements of $\mathfrak{\tilde{M}}$ with respect to the quasiorder
$\preceq_{w}$\emph{restricted} to $\mathfrak{\tilde{M}}$; these will be
referred to as \emph{strongly} minimal mechanisms of $\mathfrak{\tilde{M}.}$

\begin{theorem}
\label{Emergence of money}If $m>3$ then the three special mechanisms are
precisely the strongly minimal mechanisms of both $\mathfrak{M}_{\ast}\left(
m\right)  $ and $\mathfrak{M}_{g}\left(  m\right)  $. Their complexities are
\[%
\begin{tabular}
[c]{|c|c|c|c|}\hline
& Star & Cycle & Complete\\\hline
$\pi(M\mathcal{)}$ & $4$ & $2$ & $m(m-1)$\\\hline
$\tau(M)$ & $2$ & $m-1$ & $1$\\\hline
\end{tabular}
\ \ \ \ \ \ \ \ \ \
\]

\end{theorem}

The array clearly exhibits the superiority of the star mechanism. As the
number of commodities $m$ increases, the other two will beat star slightly in
one component, but will lose by a huge margin to star in the other component,
with the upshot that the star is the overall winner:

\begin{theorem}
\label{Star mechanism} For any strictly positive $\lambda$ and $\mu,$ there
exists $m_{0}$ such that the star mechanism is the unique maximizer of
$\lambda\pi(M)+\mu\tau(M)$ on $\mathfrak{M}_{\ast}\left(  m\right)  $ and on
$\mathfrak{M}_{g}\left(  m\right)  $ whenever $m\geq m_{0}.$
\end{theorem}

In the star mechanism\footnote{The star mechanism (also known as the
\textquotedblleft Shapley-Shubik mechanism\textquotedblright, see
\cite{Shapley: 1976}, \cite{Shapley-Shubik: 1977}, \cite{Shubik: 1973}) has
been much-studied in the literature in different contexts, see, \emph{e.g.,
}\cite{Dubey-Shubik:1978}, \cite{Dubey-MasColell-Shubik:1980},
\cite{Dubey-Shapley: 1994}, \cite{Dubey-Geanakoplos: 2003},
\cite{Dubey-Geanakoplos ISLM}, \cite{Peck:1992}, \cite{Peck-Shell:1992},
\cite{Postlewaite:1978}, , \cite{Shubik:1993}, \cite{Shubik-Wilson},
\cite{Shubik:1993}, \cite{starr:2012}. The complete mechanism (also known as
\textquotedblleft Shapley's windows mechanism\textquotedblright) is analysed
in \cite{Sahi-Yao:1989}. All other $G$-mechanisms are discussed in
\cite{Dubey-Sahi: 2003}.}, the\emph{ pair }of edges $im,mi$ constitutes a
\emph{bilateral market} between $i$ and $m$ for all $i\neq m$. Thus the
central node $m$ plays the role of money, mediating trade between various
markets. Furthermore these markets are \emph{decentralized }in that the trade
at any market is independent of the offers at all \emph{other} markets.

\section{Proof of Theorem \ref{Emergence of Prices}}

We fix a mechanism $M$ in $\mathfrak{M}$ with net trade function $\nu\left(
a,b\right)  $. Consider the set of pairs $\left(  i,j\right)  $ for which
there is at least one pure $ij$-index in $K$, and fix a subset $P\subset K$
which contains \emph{exactly} one $ij$-index for each such pair. Let
$S_{P}\subset S$ denote the set of $P$-offers, i.e. those $a$ satisfying
$a_{h}=0$ for $h\notin P$, and further define the set of $P$-offers
\textquotedblleft subordinate\textquotedblright\ to $v$ as follows:
\[
S_{P}\left(  v\right)  =\left\{  a\in S_{P}:\overline{a}\leq v\right\}
\]
Given a vector $v\in S$ we write $\left\langle v\right\rangle $ for the class
of vectors with the same sign as $v$, thus $w\in\left\langle v\right\rangle $
if each component $w_{i}$ has the same sign $(+,-,0)$ as $v_{i}$.

\begin{lemma}
\label{lem:vw}Let $v,w\in S$ then the following are equivalent.

\begin{enumerate}
\item There is an $a\in S_{P}\left(  v\right)  $ such that $v+\nu\left(
a,b\right)  \in\left\langle w\right\rangle $ for some $b\in S_{+}$

\item There is an $a\in S_{P}\left(  v\right)  $ such that $v+\nu\left(
a,b\right)  \in\left\langle w\right\rangle $ for all $b\in S_{+}$

\item For each $u\in\left\langle v\right\rangle $ there is an $a\in
S_{P}\left(  u\right)  $ such that $u+\nu\left(  a,b\right)  \in\left\langle
w\right\rangle $ for all $b\in S_{+}$
\end{enumerate}
\end{lemma}

\begin{proof}
It is evident that (3) implies (2), and (2) implies (1). We now show that (1)
implies (3). Suppose $v,a,b,w$ satisfy (1). Given $u\in\left\langle
v\right\rangle $ and $b_{\ast}\in S_{+}$, we need to find $a_{\ast}\in
S_{P}(u)$ such that $u,a_{\ast},b_{\ast},w$ satisfy (3). Since $u$ and $v$
have the same signs there exist positive scalars $\lambda_{i}$ such that
$u_{i}=\lambda_{i}v_{i}$ for all $i.$ Define $a_{\ast}$ by $\left(  a_{\ast
}\right)  _{i}=\lambda_{i}a_{i}$, where (recall) $a_{i}$ is the vector
obtained from $a$ by restricting to the $K_{i}$-components. Now we have%
\begin{align*}
v+\nu\left(  a,b\right)   &  =\left(  v-\overline{a}\right)  +r\left(
a,b\right) \\
u+\nu\left(  a_{\ast},b_{\ast}\right)   &  =\left(  u-\overline{a_{\ast}%
}\right)  +r\left(  a_{\ast},b_{\ast}\right)
\end{align*}
By construction of $a_{\ast}$ we have $\left(  v-\overline{a}\right)
_{i}=\lambda_{i}\left(  u-\overline{a_{\ast}}\right)  _{i}$ for all $i$, and
hence $\left\langle v-\overline{a}\right\rangle =\left\langle u-\overline
{a_{\ast}}\right\rangle $. Also since $a$ and $a_{\ast}$ are $P$-offers, by
\emph{Aggregation }and \emph{Invariance} we have $\left\langle r\left(
a,b\right)  \right\rangle =\left\langle r\left(  a,b_{\ast}\right)
\right\rangle =\left\langle r\left(  a_{\ast},b_{\ast}\right)  \right\rangle
$. We note that if $x,y$ are non-negative vectors then $\left\langle
x+y\right\rangle $ is uniquely determined by $\left\langle x\right\rangle $
and $\left\langle y\right\rangle $, thus we get%
\[
\left\langle u+\nu\left(  a_{\ast},b_{\ast}\right)  \right\rangle
=\left\langle v+\nu\left(  a,b\right)  \right\rangle =\left\langle
w\right\rangle
\]
which establishes (3).
\end{proof}

We note that Lemma \ref{lem:vw} (3) only depends on $\left\langle
v\right\rangle $ and $\left\langle w\right\rangle $ and we will write
$\left\langle v\right\rangle \rightarrow\left\langle w\right\rangle $ if it holds.

\begin{lemma}
For any $\left(  a,b\right)  \in S\times S_{+}$ there is $a_{\ast}\in
S_{P}\left(  \overline{a}\right)  $ such that%
\begin{equation}
\left\langle r\left(  a,b\right)  \right\rangle =\left\langle \overline{a}%
+\nu\left(  a_{\ast},b\right)  \right\rangle . \label{a*}%
\end{equation}

\end{lemma}

\begin{proof}
By \emph{Aggregation}, it suffices to prove this when $a$ is a $K_{i}$-offer
for some $i$. By \emph{Flexibility} there is some $a_{\ast}\in S_{P}\left(
\overline{a}\right)  $ such that $r_{i}\left(  a_{\ast},b\right)  =0$, while
$r_{j}\left(  a_{\ast},b\right)  $ has the same sign as $r_{j}\left(
a,b\right)  $ for all $j\neq i$. We write
\[
\overline{a}+\nu\left(  a_{\ast},b\right)  =\left(  \overline{a}%
-\overline{a_{\ast}}\right)  +r\left(  a_{\ast},b\right)
\]
and note that since $a_{\ast}$ is a pure $K_{i}$-offer, the sign of $r\left(
a_{\ast},b\right)  $ does not change if we rescale $a_{\ast}$. If
$r_{i}\left(  a,b\right)  =0$ we scale up $a_{\ast}$ to ensure $\overline
{a_{\ast}}=\overline{a}$, while if $r_{j}\left(  a,b\right)  >0$ then we scale
down $a_{\ast}$ to ensure $\overline{a_{\ast}}\lvertneqq\overline{a}$; in each
case the rescaled $a_{\ast}$ satisfies (\ref{a*}).
\end{proof}

\begin{lemma}
\label{lm:pure}$v^{1}\rightarrow\cdots\rightarrow v^{t}$ implies $\left\langle
v^{1}\right\rangle \rightarrow\cdots\rightarrow\left\langle v^{t}\right\rangle
$.
\end{lemma}

\begin{proof}
It suffices to show that $v\rightarrow w$ implies $\left\langle v\right\rangle
\rightarrow\left\langle w\right\rangle $. Now by definition
\[
w=v+\nu\left(  a,b\right)  \text{ for some }\left(  a,b\right)  \in S\times
S_{+}\text{ with }\overline{a}\leq v.
\]
If $a_{\ast}$ is as in (\ref{a*}) then the identities
\begin{align*}
v+\nu\left(  a_{\ast},b\right)   &  =\left(  v-\overline{a}\right)  +\left(
\overline{a}+\nu\left(  a_{\ast},b\right)  \right) \\
v+\nu\left(  a,b\right)   &  =\left(  v-\overline{a}\right)  +r\left(
a,b\right)
\end{align*}
imply $\left\langle v+\nu\left(  a_{\ast},b\right)  \right\rangle
=\,\left\langle w\right\rangle $, whence $\left\langle v\right\rangle
\rightarrow\left\langle w\right\rangle $ by Lemma \ref{lem:vw} (1).
\end{proof}

\begin{proposition}
\label{Convertibility} For $b\in S_{+}$ and any $i\neq j$ there is $a\in
S_{P}$ such that $\nu(a,b)$ is an $\bar{\imath}j$-vector.
\end{proposition}

\begin{proof}
Let $v$ be an $i$-vector and let $t=\tau_{ij}\left(  M\right)  $ then by
definition we have a sequence%
\[
v\rightarrow v^{1}\rightarrow\cdots\rightarrow v^{t-1}=w
\]
where $w$ is a $j$-vector. By the previous lemma we get%
\[
\left\langle v\right\rangle \rightarrow\left\langle v^{1}\right\rangle
\rightarrow\cdots\rightarrow\left\langle v^{t-1}\right\rangle \rightarrow
\left\langle w\right\rangle
\]
By Lemma \ref{lem:vw} (3) this means we can find sequences%
\[
u^{i}\in\left\langle v^{i}\right\rangle ,a^{i}\in S_{P}\left(  u^{i}\right)
\text{ for }i=0,\ldots,t-1
\]
such that $u^{i}+\nu\left(  a^{i},b\right)  =u^{i+1}$. If $a=\sum a^{i}$ then
we have $a\in S_{P}$ and%
\[
\nu\left(  a,b\right)  =%
{\textstyle\sum}
\nu\left(  a^{i},b\right)  =u^{t}-u^{1}%
\]
which is an $\bar{\imath}j$-vector.
\end{proof}

It will be convenient to write an $\bar{\imath}j$-vector in the form $\left(
-x,y\right)  $ after suppressing the other components. In the context of the
above proposition if $\nu\left(  a,b\right)  =\left(  -x,y\right)  $ then by
linearity $\nu\left(  a/x,b\right)  =\left(  -1,y/x\right)  $, and we will say
that the offer $a$ (or $a/x$) achieves an $ij$-exchange ratio of $y/x$ at $b$.

\begin{lemma}
If $a^{\prime},a^{\prime\prime}$achieve $ij$-exchange ratios $\alpha^{\prime
},\alpha^{\prime\prime}$ at $b$, then $\alpha^{\prime}=\alpha^{\prime\prime}$.
\end{lemma}

\begin{proof}
By the previous proposition there exists an $a$ such that $\nu\left(
a,b\right)  $ is a $\bar{j}i$-vector; if $\alpha$ is the corresponding
exchange ratio then by rescaling $a,a^{\prime},a^{\prime\prime}$ we may assume
that
\[
\nu\left(  a,b\right)  =\left(  1,-\alpha\right)  ,\nu\left(  a^{\prime
},b\right)  =\left(  -1,\alpha^{\prime}\right)  ,\nu\left(  a^{\prime\prime
},b\right)  =\left(  -1,\alpha^{\prime\prime}\right)  .
\]
By Proposition \ref{Linearity} we get
\[
\nu\left(  a+a^{\prime},b\right)  =\left(  0,\alpha-\alpha^{\prime}\right)
\]
Now by \emph{Non-dissipation} we get $\alpha\geq\alpha^{\prime}$, and
exchanging the roles of $i$ and $j$ we conclude that $\alpha^{\prime}%
\geq\alpha$ and hence that $\alpha=\alpha^{\prime}$. Arguing similarly we get
$\alpha=\alpha^{\prime\prime}$ and hence that $\alpha^{\prime}=\alpha
^{\prime\prime}$.
\end{proof}

\begin{proof}
[Proof of Theorem \ref{Emergence of Prices}]Fix $b\in S_{+}$ and consider the
vector%
\[
p=\left(  1,p_{2},\ldots,p_{m}\right)
\]
where $p_{j}^{-1}$ is the $1j$-exchange ratio at $b$, as in the previous
lemma. We will show that $p$ satisfies the conditions of Theorem
\ref{Emergence of Prices}, \textit{i.e.} that
\begin{equation}
p\cdot\nu\left(  a,b\right)  =0\text{ for all }a. \label{=p.nu}%
\end{equation}
We argue by induction on the number $d\left(  a,b\right)  $ of non-zero
components of $\nu\left(  a,b\right)  $ in positions $2,\ldots,m$. If
$d\left(  a,b\right)  =0$ then $\nu\left(  a,b\right)  =0$ by
\emph{Non-dissipation} and (\ref{=p.nu}) is obvious. If $d\left(  a,b\right)
=1$ then $\nu\left(  a,b\right)  $ is either an $\bar{1}j$-vector or a
$\bar{j}1$ vector, which by the definition of $p_{j}$ and the previous lemma
is necessarily of the form%
\[
\left(  -x,xp_{j}^{-1}\right)  \text{ or }\left(  x,-xp_{j}^{-1}\right)  ;
\]
for such vectors (\ref{=p.nu}) is immediate. Now suppose $d\left(  a,b\right)
=d>1$ and fix $j$ such that $\nu_{j}\left(  a,b\right)  \neq0$. Then we can
choose $a^{\prime}$ such that $\nu\left(  a^{\prime},b\right)  $ is a $\bar
{1}j$ or a $\bar{j}1$- vector such that $\nu_{j}\left(  a,b\right)  =-\nu
_{j}\left(  a^{\prime},b\right)  .$ It follows that $d\left(  a+a^{\prime
},b\right)  <d$ and by linearity we get
\[
p\cdot\nu\left(  a,b\right)  =p\cdot\nu\left(  a+a^{\prime},b\right)
-p\cdot\nu\left(  a^{\prime},b\right)  .
\]
By the inductive hypothesis the right side is zero, hence so is the left side.

Finally the uniqueness of the price function is obvious, because the return
function of the mechanism dictates how many units of $j$ may be obtained for
one unit of $i$, yielding just one possible candidate for the exchange rate
for every pair $ij.$
\end{proof}

\section{Proof of Theorem \ref{Emergence of G-mechanisms}}

We say a matrix $X$ is an $S\times T$ matrix if its rows and columns are
indexed by finite sets $S$ and $T$ respectively; if $Y$ is a $T\times U$
matrix then the product $XY$ is a well-defined $S\times U$ matrix. For the set
$\left[  n\right]  =\left\{  1,\ldots,n\right\}  $ we will speak of $n\times
T$ matrices instead of $\left[  n\right]  \times T$ matrices, etc.

Let $M\in\mathfrak{M(}m\mathfrak{)}$ and write $K_{i}=K_{i}\left(  M\right)  $
and $K=\coprod_{i}K_{i}$ as usual. For any vector $v\in\mathbb{R}_{++}^{m}$,
let $D_{v}$ denote the $m\times m$ diagonal matrix \emph{diag}$\left\{
v_{1},\ldots,v_{m}\right\}  ,$ and let $E_{v}$ denote the\ $K\times K$
\textquotedblleft extended\textquotedblright\ diagonal matrix whose $K_{i}%
$-diagonal entries are all $v_{i}$. Also let $A$ be the $m\times K$
\textquotedblleft auxiliary\textquotedblright\ matrix whose $K_{1}$-columns
are $(1,0,\ldots,0)^{t}$, $K_{2}$-columns are $(0,1,0,\ldots,0)^{t}$, etc.

\begin{lemma}
\label{lm:Nb}$M$ is uniquely determined by a map $b\mapsto N_{b}$ from $S_{+}$
to the space of non-negative $m\times K$ column-stochastic matrices as follows.

\begin{enumerate}
\item The price ray $p=p(b)$ is obtained as the unique solution of
\begin{equation}
C_{b}p=\Delta_{b}p \label{=Delp}%
\end{equation}

where $\Delta_{b}=$ $AD_{b}A^{t}$ is the diagonal matrix of column sums of
$C_{b}=N_{b}D_{b}A^{t}$

\item The return function is given by
\begin{equation}
r\left(  a,b\right)  =M_{b}a\text{ where }M_{b}=D_{p}^{-1}N_{b}E_{p}
\label{=Nb}%
\end{equation}

\end{enumerate}
\end{lemma}

\begin{proof}
Let $p=p(b)$ be the price function whose existence is guaranteed by Theorem
\ref{Emergence of Prices}. We will first prove formula \ref{=Nb} for $r(a,b)$
and then prove formula \ref{=Delp}. By Proposition \ref{Linearity}, the return
function of the mechanism $M$ is of the form $r\left(  a,b\right)  =M_{b}a$,
where $b\mapsto M_{b}$ is a map from $S_{+}$ to the space of non-negative
$m\times k$ matrices satisfying
\[
M_{b}b=Ab
\]
and the identity%
\begin{equation}
M_{E_{v}b}=D_{v}M_{b}E_{v}^{-1}\text{ for all }v\in R_{++}^{m}. \label{Mb}%
\end{equation}
(The non-negativity $M_{b}$ follows from that of $r\left(  a,b\right)  $. The
first display holds by conservation of commodities and the second by
\emph{Invariance}.) Define
\[
b^{\prime}=E_{p}b\text{, }N_{b}=M_{b^{\prime}}\text{.}%
\]
By \emph{Invariance} it follows that $p(b^{\prime})=\mathbf{1}$. Also each
column of $N_{b}=M_{b^{\prime}}$ is the return to the offer of a single unit
in some commodity. Since all prices are $1$ at $b^{\prime}$, Theorem
\ref{Emergence of Prices} implies that each column of $N_{b}$ sums to $1$,
\textit{i.e. }$N_{b}$ is column stochastic. Now by (\ref{Mb}) we get
\[
N_{b}=M_{E_{p}b}=D_{p}M_{b}E_{p}^{-1}\text{,}%
\]
whence $M_{b}=D_{p}^{-1}N_{b}E_{p}$ as desired

Now combining (\ref{=Nb}) and $M_{b}b=Ab,$ with the identity $D_{p}A=AE_{p}$
we have%
\[
N_{b}E_{p}b=D_{p}M_{b}b=D_{p}Ab=AE_{p}b.
\]
Using the identity $E_{p}b=D_{b}A^{t}p$ we can rewrite this as
\[
N_{b}D_{b}A^{t}p=AD_{p}A^{t}b,
\]
which is precisely (\ref{=Delp}).
\end{proof}

\begin{lemma}
\label{lm:span} Let $N_{b}$ in Lemma \ref{lm:Nb} and let $h\in K_{i}$ be a
pure $ij$-index.

\begin{enumerate}
\item The $h$-th column of $N_{b}$ is the $j$-th unit vector $e_{j}$,
\emph{independent} of $b$.

\item Every $K_{i}$-column of $N_{b}$ is a linear combination of the
\textquotedblleft pure\textquotedblright\ $K_{i}$-columns.
\end{enumerate}
\end{lemma}

\begin{proof}
By definition there is an $h$-offer $a$ such that $r\left(  a,b\right)
=M_{b}a$ is a $j$-return vector. This means that the $h$-th column of $M_{b}$
has a non-zero entry only in its $j$-th component. Since $N_{b}$ is obtained
from $M_{b}$ by rescaling entries this is also true of $N_{b}$. By column
stochasticity the $h$-th column of $N_{b}$ must be $e_{j}$.

For the second part, let $h^{\prime}\in K_{i}$ be an $i$-index, let $v,w$ be
the $h^{\prime}$-th columns of $N_{b}$ and $M_{b}$, and suppose the $j$-th
component of $v$ (and hence of $w$) is non-zero. It suffices to show that in
this case the mechanism has a pure $ij$-index. However if $a$ is an
$h^{\prime}$-offer then $r\left(  a,b\right)  =M_{b}a$ is a multiple of $w$,
and thus the assertion follows from \emph{Flexibility}.
\end{proof}

Let $G$ be the graph in which we connect $i$ to $j$ if $M$ has a pure
$ij$-index. Since $M$ is connected, Lemma \ref{lm:pure} implies that $G$ is
connected, and we let $M^{\prime}=M_{G}$ denote the corresponding
$G$-mechanism. We will identify the $i$-indices $K_{i}^{\prime}$ of
$M^{\prime}$ as a subset of $K_{i}$. If $M$ has several pure $ij$-indices for
a given $j$ then this involves a choice, however the choice will play no role
in the subsequent discussion. We will refer to $M^{\prime}$ as the embedded
$G$-mechanism of $M$.

To continue we need a result from \cite{Sahi:2013}. Let $G$ be any connected
directed graph on $\left\{  1,\ldots,n\right\}  $ with weights $z_{ij}$
attached to edges $ij\in G$. We write $Z=\left(  z_{ij}\right)  $ for the
$n\times n$ matrix of edge weights of $G$, setting $z_{ij}=0$ if $ij\notin G$.
We also define
\[
\delta_{j}=%
{\textstyle\sum\nolimits_{i}}
z_{ij},\quad\Delta_{Z}=diag\left(  \delta_{1},\ldots,\delta_{n}\right)  ,
\]
so that $\Delta$ is the diagonal matrix of column sums of $Z$. We define the
weight of a subgraph $\Gamma$ to be the product of its edge weights, thus
\[
w_{\Gamma}\left(  z\right)  =\prod\nolimits_{ij\in E_{\Gamma}}z_{ij}.
\]
We define an $i$\emph{-tree} in $G$ to be a (directed) subgraph $T$ with $n$
vertices and $n-1$ edges, and the futher property that $T$ contains a path
from $j$ to $i$ for every $j\neq i$. We write $\mathcal{T}_{i}$ for the set of
$i$-trees in $G$, and define%
\[
w_{i}=\sum\nolimits_{\Gamma\in\mathcal{T}_{i}}w_{\Gamma}\left(  z\right)
,\quad w=\left(  w_{1},\ldots,w_{n}\right)  ^{t}.
\]
The following lemma from \cite{Sahi:2013} is critical and paves the way for
the rest of the analysis.

\begin{lemma}
\label{XDw}If $Z,\Delta_{Z},w$ are as above then one has $Zw=\Delta_{Z}w.$
\end{lemma}

We can now prove a key property of embedded $G$-mechanisms.

\begin{proposition}
If a price ratio depends on some variable in $M^{\prime}$, then it does so in
$M$.
\end{proposition}

\begin{proof}
The pure columns of $N_{b}$ are fixed unit vectors, independent of $b$. By
assumption there is a bijection between the pure variables and the nonzero
entries $c_{rs}\left(  b\right)  $ of the matrix $C_{b}$. We denote the pure
components of $b$ by $x=\left(  x_{rs}\right)  $ and the remaining mixed
components by $y=\left(  y_{k}\right)  $. Then by the definition of $C_{b}$ we
have an expression of the form%
\begin{equation}
c_{rs}\left(  b\right)  =x_{rs}+\sum\nolimits_{k}\varepsilon_{k}\left(
b\right)  y_{k}\text{; }0\leq\varepsilon_{k}\left(  b\right)  \leq1.
\label{=crs}%
\end{equation}
By formula (\ref{=Delp}) and Lemma \ref{XDw}, the prices $p$ in $M$ and
$M^{\prime}$ are weighted sums of trees with edge weights $c_{rs}$ and
$x_{rs}$ repectively. Let $p\left(  x,y\right)  $ denote the price vector in
$M$ at $b=\left(  x,y\right)  $ and let $p\left(  x\right)  $ denote the price
vector in $M^{\prime}$ at $x$. Then by (\ref{=crs}) we get
\[
p\left(  x\right)  =\lim_{y\rightarrow0}p\left(  x,y\right)  .
\]

We now fix a pair of commodities $i,j$ and let $\pi\left(  x,y\right)  $ and
$\pi\left(  x\right)  $ denote the price ratios $p_{i}/p_{j}$ in $M$ and
$M^{\prime}$ respectively, then we have
\[
\pi\left(  x\right)  =\lim_{y\rightarrow0}\pi\left(  x,y\right)  .
\]
Thus if $\pi\left(  x\right)  $ depends on some $x$-component, so must
$\pi\left(  x,y\right)  $.
\end{proof}

\begin{proof}
[Proof of Theorem \ref{Emergence of G-mechanisms}]By lemma \ref{lm:pure},
lemma \ref{lm:span} and the previous proposition (respectively), we have:
\[
\tau_{ij}\left(  M^{\prime}\right)  =\tau_{ij}\left(  M\right)  ,\text{
}k\left(  M^{\prime}\right)  \leq k\left(  M\right)  ,\text{ }\pi_{ij}\left(
M^{\prime}\right)  \leq\pi_{ij}\left(  M\right)
\]
If $M$ is minimal then equality must hold throughout. Hence we get $k\left(
M^{\prime}\right)  =k\left(  M\right)  $ and so $M=M^{\prime}$ is a
$G$-mechanism.$.$
\end{proof}

\section{Proof of Theorem \ref{Emergence of money}}

First let us recall some basic order-theoretic notions.

\begin{definition}
\label{qo}A \emph{quasiorder} $\precsim$ on a set $X$ is a binary relation
that is reflexive ($x\precsim x)$ and transitive:%
\[
x\precsim y,y\precsim z\implies x\precsim z.
\]

We write $x\prec y$ if $x\precsim y$ holds but $y\precsim x$ does not hold. We
say that $x$ is $\precsim$-\emph{minimal} if there is no $y$ in $X$ such
$y\prec x$, equivalently if for all $y\in X$
\[
y\precsim x\implies x\precsim y.
\]

We write $X_{\prec}$ for the set of $\precsim$-minimal elements of $X$. We say
that $\precsim$ is a \emph{well}-quasiorder (wqo) if there does not exist an
infinite descending chain%
\[
\cdots\prec x_{n}\prec\cdots\prec x_{2}\prec x_{1}.
\]

\end{definition}

Note that if $\precsim$ is a wqo on $X$ and $Y\subset X$ then the restriction
of $\precsim$ defines a wqo on $Y$. In general the minimal elements $Y_{\prec
}$ can be quite different from $X_{\prec}$, however we have the following
elementary result.

\begin{lemma}
\label{wqo}If $\left(  X,\precsim\right)  $ is a wqo and $X_{\prec}\subset
Y\subset X$ then $X_{\prec}=Y_{\prec}$.
\end{lemma}

\begin{proof}
Any minimal element of $X$ that happens to lie in $Y$ is clearly minimal in
$Y$. Thus $X_{\ast}\subset Y$ implies $X_{\prec}\subset Y_{\prec}$. On the
other hand if $z$ is a non-minimal element of $X$ then $x\prec z$ for some
$x\in X_{\prec},$ otherwise we could construct an infinite descending chain in
$X$ starting with $z$. In particular any $z\in Y\setminus X_{\prec}$ satisfies
$x\prec z$ for some $x\in X_{\prec}\subset Y$, hence $z$ is not minimal in
$Y.$
\end{proof}

It is easy to check that both the quasiorders $\preceq$ and $\preceq_{w}$ that
we have introduced on $\mathfrak{M}$ are wqo's; and therefore in the proof
below we will apply the previous lemma to them.

\begin{proof}
[Proof of Theorem \ref{Emergence of money}]Let $\mathfrak{S}$ denote the set
consisting of the three special mechanisms. We need to show that $\left(
\mathfrak{M}_{g}\right)  _{\prec_{w}}=\mathfrak{S}$ and $\left(
\mathfrak{M}_{\ast}\right)  _{\prec_{w}}=\mathfrak{S}$.

It is shown in \cite{Dubey-Sahi-Shubik}\ that $\left(  \mathfrak{M}%
_{g}\right)  _{\prec_{w}}=\mathfrak{S}$.

We now prove $\left(  \mathfrak{M}_{\ast}\right)  _{\prec_{w}}=\mathfrak{S}$.
Since $\left(  \mathfrak{M}_{g}\right)  _{\prec_{w}}=\mathfrak{S}$, by Lemma
\ref{wqo} applied to the wqo $\preceq_{w}$it suffices to show that
$\mathfrak{S\subset M}_{\ast}$. We further note that
\begin{equation}
\mathfrak{M}_{\ast}=\left(  \mathfrak{M}_{g}\right)  _{\prec}. \label{Mg*}%
\end{equation}
Indeed $\mathfrak{M}_{\ast}=\mathfrak{M}_{\prec}$ by definition, and
$\mathfrak{M}_{\ast}\subset\mathfrak{M}_{g}$ by Theorem
\ref{Emergence of G-mechanisms}; now (\ref{Mg*}) follows from Lemma \ref{wqo}
applied to the wqo $\preceq$. Thus it suffices to prove that
\begin{equation}
\mathfrak{S}\subset\left(  \mathfrak{M}_{g}\right)  _{\prec} \label{Mg<}%
\end{equation}
\textit{i.e.,} that each of the three special mechanisms is $\preceq$-minimal
in $\mathfrak{M}_{g}$.

The $\preceq$-minimality is obvious for the complete graph since any other
graph would have some $\tau_{ij}>1$, and also for the cycle since any other
graph would have some $k_{i}>1$. To establish $\preceq$-minimality for the
star graph, it suffices to show that any non-star graph $G$ has either some
$\pi_{ij}\geq5$ or some $\tau_{ij}\geq3$. For this we note that $\pi\geq5$
holds by Theorem 16 of \cite{Dubey-Sahi-Shubik} if $G$ is not a rose or a
chorded cycle; while $\tau\geq3$ holds trivially for non-star roses and by
Lemma 37 of \cite{Dubey-Sahi-Shubik} for chorded cycles. This completes the
proof of (\ref{Mg<}) and hence of $\left(  \mathfrak{M}_{\ast}\right)
_{\prec_{w}}=\mathfrak{S.}$
\end{proof}

\begin{remark}
For $m=3$, Lemma 37 of \cite{Dubey-Sahi-Shubik} does not hold and we have an
additional strongly minimal mechanism with $\left(  \tau,\pi\right)  =\left(
2,4\right)  $, namely the chorded triangle
\[%
\begin{tabular}
[c]{|lll|}\hline
$\cdot$ &  & \\
$\downarrow$ & $\nwarrow$ & \\
$\cdot$ & $\leftrightarrows$ & $\cdot$\\\hline
\end{tabular}
\ \ \ \ \
\]

\end{remark}

\end{document}